\documentclass[oribibl]{llncs}

\newcommand{\keywords}[1]{\par\addvspace\baselineskip
\noindent\keywordname\enspace\ignorespaces#1}
\usepackage{amsfonts,amsmath,amssymb}

\usepackage{verbatim}
\usepackage{graphics,graphicx,colortbl}
\newcommand{\bea}{\begin{eqnarray}}
\newcommand{\eea}{\end{eqnarray}}
\newcommand{\be}{\begin{eqnarray*}}
\newcommand{\ee}{\end{eqnarray*}}
\newtheorem{algorithm}{Algorithm}
\begin{document}
\title{Representing Boolean Functions Using Polynomials: More Can Offer Less}
\titlerunning{Representing Boolean Functions Using Polynomials}

\author{Yi Ming Zou}
\institute{Department of Mathematical Sciences\\ University of Wisconsin-Milwaukee\\
Milwaukee, WI 53201, USA}
\toctitle{Lecture Notes in Computer Science}
\tocauthor{Authors' Instructions}
\maketitle

%%%%%%%%%%%%%%%%%%%%%%%%%%%%%%%%%%%%%%%%%%%%%%%%%%%%%%%%%%%%%%%%%%%%%%%%%%%%%%%%%%%%%%%%%%%%%%%%%%%%%%%%%%%%%%
\begin{abstract}
Polynomial threshold gates are basic processing units of an artificial neural network. When the input vectors are binary vectors, these gates correspond to Boolean functions and can be analyzed via their polynomial representations. In practical applications, it is desirable to find a polynomial representation with the smallest number of terms possible, in order to use the least possible number of input lines to the unit under consideration. For this purpose, instead of an exact polynomial representation, usually the sign representation of  a Boolean function is considered. The non-uniqueness of the sign representation allows the possibility for using a smaller number of monomials by solving a minimization problem. This minimization problem is combinatorial in nature, and so far the best known deterministic algorithm claims the use of at most $0.75\times 2^n$ of the $2^n$ total possible monomials. In this paper, the basic methods of representing a Boolean function by polynomials are examined, and an alternative approach to this problem is proposed. It is shown that it is possible to use at most $0.5\times 2^n = 2^{n-1}$ monomials based on the $\{0, 1\}$ binary inputs by introducing extra variables, and at the same time keeping the degree upper bound at $n$. An algorithm for further reduction of the number of terms that used in a polynomial representation is provided. Examples show that in certain applications, the improvement achieved by the proposed method over the existing methods is significant.  
\end{abstract}
%%%%%%%%%%%%%%%%%%%%%%%%%%%%%%%%%%%%%%%%%%%%%%%%%%%%%%%%%%%%%%%%%%%%%%%%%%%%%%%%%%%%%%%%%%%%%%%
%\maketitle
\date{}

%%%%%%%%%%%%%%%%%%%%%%%%%%%%%%%%%%%%%%%%%%%%%%%%%%%%%%%%%%%%%%%%%%%%%%%%%%%%%%%%%%%%%%%%%%%%%%%%%%%%%%%%%%%%%%%%%%%%%%
%%%%%%%%%%%%%%%%%%%%%%%%%%%%%%%%%%%%%%%%%%%%%%%%%%%%%%%%%%%%%%%%%%%%%%%%%%%%%%%%%%%%%%%%%%%%%%%%%%%%%%%%%%%%%%%%%%%%%%
\section{Introduction}
\par
In this paper, we consider the problem of using a smaller number of monomial terms in the expression of a polynomial threshold neuron with binary inputs \cite{Hass95}. These neurons correspond to Boolean functions, and it is known that any Boolean function can be represented by a polynomial function. There are three basic methods for representing a Boolean function by a polynomial in $n$ variables. The first one is to consider a Boolean function as a function $\{0,1\}^n\longrightarrow \{0,1\}$, and in this case the representation is unique \cite{Rud74}. The second one is to consider a Boolean function as a function $\{-1,1\}^n\longrightarrow \{-1,1\}$, and in this case, the representation is also unique \cite{Sak93,Wang91} (see also the discussion later in this section). The third, called sign representation, which is not unique and allows flexibility in the choice of a polynomial representation, has the input set $\{-1,1\}^n$ and the output set given by all nonzero real numbers\footnote{In some literature, the sign representations are allowed to take $0$ as an output. It is easy to see that these two definitions are equivalent.}. Recall that a polynomial $p: \{-1,1\}^n\longrightarrow\mathbb{R}$ is said to be sign representing a Boolean function $f: \{-1,1\}^n\longrightarrow \{-1,1\}$ if $f = sign(p)$ for all vectors in $\{-1,1\}^n$. Since for the values $0$ and $1$ we have $x^2=x$, and for the values $\pm 1$ we have $x^2=1$, in all three cases, there are $2^n$ linearly independent monomials of the form 
\bea\label{e1}
x_{i_1}x_{i_2}\cdots x_{i_k},\;\; 1\le i_1< i_2< \cdots< i_k,\;\; 0\le k\le n,
\eea
such that any other polynomial is a linear combination of these $2^n$ monomials (when $k=0$, the corresponding monomial is $1$).
\par
Before examining these methods of polynomial representations of Boolean functions in some detail, we first consider the case $n=3$ as an example. Since the binary sequences of length $3$ are in one-to-one correspondence with the vertices of a three dimensional cube, the corresponding Boolean functions can also be specified visually using a cube. In Fig. 1, (a) shows the corresponding labeling of the vertices of a $3$-dimensional cube using vectors of $\{-1,+1\}^3$. When use the vectors in $\{0,1\}^3$, we can just change the $-1$'s to $0$'s. Fig. 1 (b) and (c) describe two Boolean functions.
\par
\begin{figure}[h]
\begin{center}
\includegraphics[width=3.4in, height=1.4in]{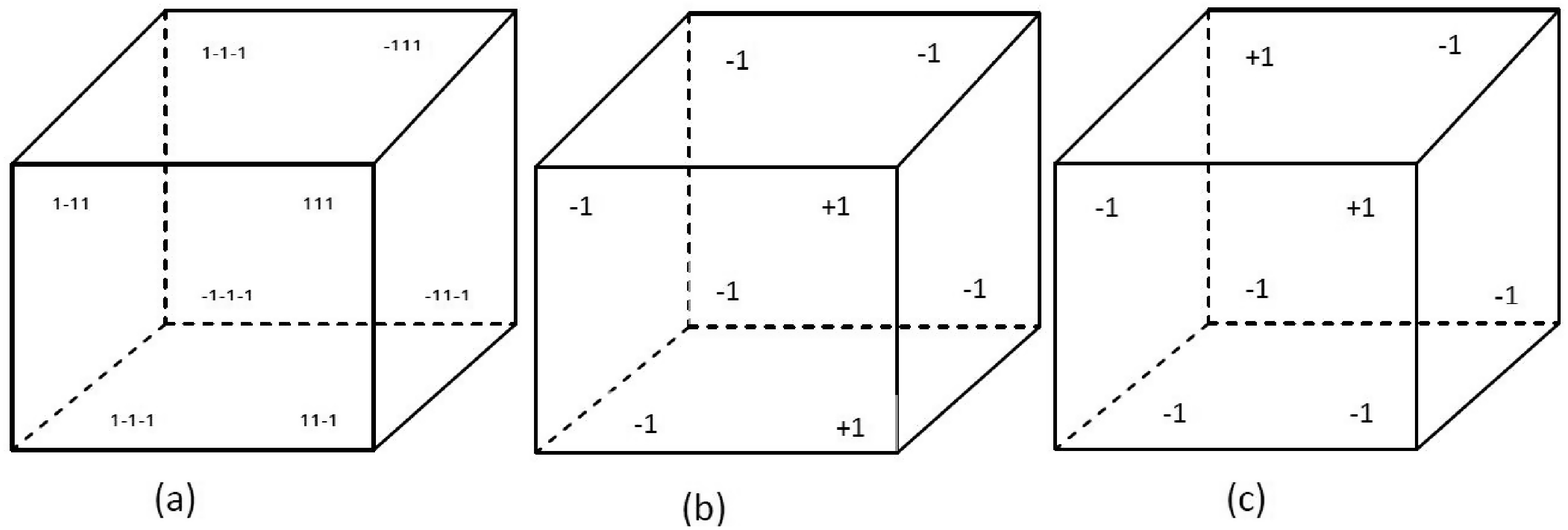} 
\caption{{\footnotesize Labeling of the cube and the descriptions of two Boolean functions. In (b) and (c), the values of the corresponding Boolean functions are specified by the numbers $+1$ or $-1$.}} \label{Fig1}
\end{center}
\end{figure}
\par
If the function given by Fig. 1 (b) is represented by a polynomial $\{-1,+1\}^3\longrightarrow \{-1,+1\}$, then the function is $f=\frac{1}{2}(-1+x+y+xy)$ (Fig. 2 (a)); if it is given by a sign representation, then we can take $f=x+y-1$ (Fig. 2 (b)); and if it is considered as a function  $\{0,1\}^3\longrightarrow \{0,1\}$, then it is given by $f=xy$ (Fig. 2 (c)). Thus for this function, viewing the function as a function $\{0,1\}^3\longrightarrow \{0,1\}$ offers the simplest representation, since it is easy to see that no single monomial can sign represent this function by checking each one of them.
\par
\begin{figure}[h]
\begin{center}
\includegraphics[width=3.4in, height=0.8in]{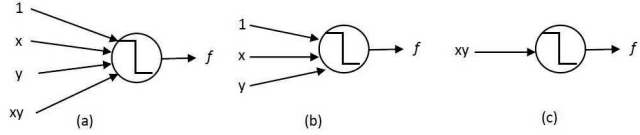} 
\caption{{\footnotesize The neurons for the function of Fig. 1 (b) under different settings. The weights of the input lines are omitted to simplify the diagrams.}} \label{Fig2}
\end{center}
\end{figure}
\par
Consider the function described by Fig. 1 (c). For a sign representation, we can take $f=-1+z+xy$; while viewing as a function $\{0,1\}^3\longrightarrow \{0,1\}$, the function is given by $f=z+xz+yz$. Thus in this case, the sign representation is simpler. This example shows that both the sign representations and the $\{0,1\}$ valued polynomial representations have advantages and disadvantages in applications.
\par
To further analyze the problem under consideration, we start with polynomial functions $\{-1,1\}^n\longrightarrow \{-1,1\}$. Since each of these polynomial functions can be expressed as a linear combination of the monomials described by (\ref{e1}), we can assign a $2^n\times 1$ column vector to each of these monomials by using the values of the given monomial on all the vectors of $\{-1,+1\}^n$. It is known that these column vectors form an orthogonal set, and if we order the monomials appropriately, then the columns actually form the Hadamard matrix $H_{2^n}$ \cite{Siu95}. 
\par
For instance, in case of $n=3$, if we order the monomials as 
\bea\label{f1}
1,z,y,yz,x,xz,xy,xyz
\eea
 and order the vectors of $\{-1,+1\}^3$ by{\footnotesize
\be
(+1+1+1), (+1+1-1),(+1-1+1),(+1-1-1),\\
(-1+1+1),(-1+1-1),(-1-1+1),(-1-1-1),
\ee}
 then the columns of the monomials in (\ref{f1}) form the Hadamard matrix $H_8$. 
\par
It is known that Hadamard matrices are symmetric and the columns of a Hadamard matrix are orthogonal, that is, the Hadamard matrix $H_{2^n}$ satisfies 
\be
H_{2^n}^T = H_{2^n},\quad H_{2^n}H_{2^n} = 2^nI_{2^n}, 
\ee
where $I_{2^n}$ is the identity matrix of size $2^n$. Thus for any $2^n\times 1$ column vector $\mathbf{b}$ with entries in $\{-1,+1\}$, the following system of linear equations
\bea\label{e2}
H_{2^n}\mathbf{x} = \mathbf{b}
\eea 
has a unique solution
\bea\label{e3}
\mathbf{x} = \frac{1}{2^n}H_{2^n}\mathbf{b}.
\eea
\par
This provides a simple method for finding the polynomial representation of a polynomial function $f: \{-1,1\}^n\longrightarrow \{-1,1\}$ given the values of the function. The advantage of using this polynomial representation is that it is rather simple to compute the representing polynomial because of the nice property of the Hadamard matrices. The disadvantage is that the representation is unique and the number of terms used in the representation cannot be reduced. 
\par
On the other hand, the sign representation allow more flexibility. For a sign representation, we can replace the vector $\mathbf{b}$ in equation (\ref{e3}) by any column vector $\mathbf{c}$ such that the signs of its entries match those of $\mathbf{b}$ (we say that $\mathbf{c}$ sign represents $\mathbf{b}$), that is, if $\mathbf{c}=(c_1,c_2,\ldots,c_{2^n})^T$ and $\mathbf{b}=(b_1,b_2,\ldots,b_{2^n})^T$, then $sign(c_i) = b_i,i=1,2,...,2^n$. If $\mathbf{c}$ sign represents $\mathbf{b}$, then the corresponding solution of (\ref{e3}) will provide a sign representation for the Boolean function $f$. Consider the system of equations
\bea\label{e4}
H_{2^n}\mathbf{x} = \mathbf{c}.
\eea
If we let $D_{\mathbf{b}}$ be the diagonal matrix with the diagonal entries given by the entries of $\mathbf{b}$, then $\mathbf{c}' :=D_{\mathbf{b}}\cdot\mathbf{c}$ has all positive entries. We write $\mathbf{c}'>\mathbf{0}$ if the entries of $\mathbf{c}'$ are all positive. 
\par
With these notation, we can formulate the problem of finding a polynomial sign representation with the smallest number of terms as the following optimization problem:
\bea\label{e5}
\mbox{Min}||\mathbf{x}||_{0} \quad\mbox{subject to}\quad D_{\mathbf{b}}H_{2^n}\mathbf{x} > \mathbf{0},
\eea
where the $\ell_0$ norm counts the nonzero entries of $\mathbf{x}$. 
\par
This minimization problem offers a rich theory and it is related to complexity problems \cite{Anth95, Bei94, Bru90, ODon08, Sak93}. Yet, in spite of all the attentions this problem has received \cite{Asp94, Boho02, EK01, Mat10, Nis94, Par02, SchM98}, there was no deterministic algorithm toward this minimization problem until the recent work \cite{Ozt09}, where a deterministic algorithm claims the use of at most $0.75\times 2^n$ monomials. However, as pointed out in \cite{Ozt09}, the algorithm proposed there is computationally costly, in particular, comparing with the simplicity of finding the exact solution via (\ref{e3}). 
\par
In contrary, due to the fact that when a Boolean function is considered as a polynomial function $\{0,1\}^n\longrightarrow \{0,1\}$ the expression is unique, relatively little attention has been paid to this case. We have seen from the examples we considered before that in some cases, these representations actually offer better solutions. Consider the case $n=3$ again. The polynomial function $\{0,1\}^3\longrightarrow \{0,1\}$ with the maximum number of terms is 
\bea\label{e6}
f=1+x+y+z+xy+xz+yz+xyz, 
\eea
which takes value $1$ on $(0,0,0)$ and $0$ for all other inputs. This polynomial function uses all $8$ possible terms and cannot be reduced. Since $\{0,1\}$ inputs and outputs are natural for Boolean functions, Boolean functions of the form $\{0,1\}^n\longrightarrow \{0,1\}$ are widely used in applications ranging form computer science to modeling biological systems. Thus we would like to ask the question of whether one can do better using this type of polynomial representations in neural networks. In the next section, we will explain how to represent the functions like the one in (\ref{e6})  by introducing extra variables and thus reducing the number of terms used while keeping the degree upper bound to be $n$.
%%%%%%%%%%%%%%%%%%%%%%%%%%%%%%%%%%%%%%%%%%%%%%%%%%%%%%%%%%%%%%%%%%%%%%%%%%%%%%%%%%%%%%%%%%%%%%%%%%%%%%%%%%%%%%%%%%%%%%%%%%%%%%%%%%%%%%%%%%%%%%%%%%%%%%%%
\section{Main Result}
\par
Recall the polynomial function $f$ defined by (\ref{e6}). Observe that it can be factored as $f = (x+1)(y+1)(z+1)$.
Suppose we introduce three extra variables $u$, $v$, and $w$, such that $u=x+1$, $v=y+1$, and $w=z+1$, then we can have a polynomial representation of $f$ with only one term: $f = uvw$, which is as simple as one can get. To put our observation on rigorous mathematical ground, we start with recalling the formal definition of the Boolean polynomial algebra in $n$ variables.
\par
To simplify our writing, let $F_2 =\{0,1\}$. Since $F_2$ is the Galois field of $2$ elements, we can consider the ring of multivariate polynomials 
\bea\label{e8}
F_2[\mathbf{x}] := F_2[x_1,x_2,...,x_n],\;\;\mathbf{x}=(x_1,x_2,\ldots,x_n),
\eea
with coefficients in $F_2$. The polynomial ring $F_2[\mathbf{x}]$ is not what we use to represent polynomial binary functions since $x_i^2\ne x_i$ in this polynomial ring. To make $x_i^2=x_i$, we need to quotient out the ideal\footnote{An ideal of $F_2[\mathbf{x}]$ is a nonempty subset $I\subset F_2[\mathbf{x}]$ such that (1) $a+b\in I,\;\forall a, b \in I$; and (2) $pa\in I,\;\forall a\in I, p\in F_2[\mathbf{x}]$.} $I$ of $F_2[\mathbf{x}]$  generated by the binomials  (note that over the Galois field $F_2$, $x_i-x_i^2 = x_i+x_i^2$): $x_i+x_i^2,\;i=1,\ldots, n$.
Explicitly,
\be
I = \{\sum_{i=1}^ng_i(x_i+x_i^2)\;|\; g_i\in F_2[\mathbf{x}],\; 1\le i\le n\}.
\ee
The quotient $F_2[\mathbf{x}]/I$ is the Boolean algebra in which $x_i=x_i^2,\;1\le i\le n$. We denote this Boolean algebra by $B[\mathbf{x}]$, where $B=\{0,1\}$, to distinguish it from $F_2[\mathbf{x}]$. We have the following well-known fact \cite{Rud74}:
\begin{proposition}
Every polynomial in $B[\mathbf{x}]$ is a linear combination (with coefficients $0$ or $1$) of the $2^n$ monomials described in (\ref{e1}) and every function $f: F_2^n\longrightarrow F_2$ can be uniquely represented by such a polynomial.
\end{proposition}
\par
Following a similar construction, we start with a polynomial ring in $2n$ variables 
\bea\label{e9}
F_2[\mathbf{x},\mathbf{y}] := F_2[x_1,\ldots,x_n,y_1,\ldots, y_n],
\eea
and form the quotient by taking the ideal $J$ of $F_2[\mathbf{x},\mathbf{y}]$ generated by the following set of binomials and trinomials:
\bea\label{e10}
x_i+x_i^2,\;\; y_i +x_i+1,\;\; 1\le i\le n.
\eea
That is, we consider the Boolean algebra
\bea\label{e11}
B[\mathbf{x}/\mathbf{y}] = F_2[\mathbf{x},\mathbf{y}]/J,
\eea
where we have used the notation ``$\mathbf{x}/\mathbf{y}$'' instead of ``$\mathbf{x},\mathbf{y}$'' in order to distinguish the algebra constructed here from the usual Boolean algebra. We have the following theorem:
\begin{theorem}\label{t1} The Boolean algebras, $B[\mathbf{x}]$, $B[\mathbf{y}]$, and $B[\mathbf{x}/\mathbf{y}]$ are all isomorphic, i.e. 
\bea\label{e12}
B[\mathbf{x}]\cong B[\mathbf{y}] \cong B[\mathbf{x}/\mathbf{y}].
\eea
\end{theorem}
\begin{proof} Since $F_2[\mathbf{x},\mathbf{y}]=F_2[\mathbf{x}][\mathbf{y}]$ and $B[\mathbf{x}]$ is a quotient of $F_2[\mathbf{x}]$, by the substitution principle \cite{Art91}, there exists a ring homomorphism
\be
\phi : F_2[\mathbf{x},\mathbf{y}]\longrightarrow B[\mathbf{x}]
\ee
defined by 
\be
\phi : x_i\longrightarrow x_i,\; y_i\longrightarrow x_i+1,\quad 1\le i\le n.
\ee
 This homomorphism is clearly onto, so by the first isomorphism theorem \cite{Art91}, we have
\bea\label{e13}
B[\mathbf{x}]\cong F_2[\mathbf{x},\mathbf{y}]/\mbox{ker}(\phi), 
\eea
where
\be
\mbox{ker}(\phi) := \{p\in F_2[\mathbf{x},\mathbf{y}]\;|\; \phi(p)=0\}.
\ee
We need to show that $\mbox{ker}(\phi) = J$. It is clear that $x_i+x_i^2,\;y_i+x_i+1,\;i=1,\ldots, n$, are all in $\mbox{ker}(\phi)$, so $\mbox{ker}(\phi) \supseteq J$. To see that they are actually equal, we note that in $B[\mathbf{x}/\mathbf{y}] = F_2[\mathbf{x},\mathbf{y}]/J$, we have $y_i=x_i+1$, so every element can be expressed as a linear combination of the $2^n$ monomials in $x_i,\;i=1,\ldots,n,$ as described in (\ref{e1}). So if $\mbox{ker}(\phi) \supsetneq J$, then we would have the following relation on the cardinalities of the Boolean algebras
\be
|F_2[\mathbf{x},\mathbf{y}]/\mbox{ker}(\phi)|\lneq |B[\mathbf{x}/\mathbf{y}]|\le |B[\mathbf{x}]|,
\ee
which contradicts (\ref{e13}). We have just proved that $B[\mathbf{x}/\mathbf{y}]\cong B[\mathbf{x}]$, the proof for $B[\mathbf{x}/\mathbf{y}]\cong B[\mathbf{y}]$ is similar. Q.E.D.
\end{proof}
\par
Theorem \ref{t1} allows us to identify the Boolean algebra $B[\mathbf{x}]$ with the Boolean algebra $B[\mathbf{x}/\mathbf{y}]$. To identify the elements of these two Boolean algebras as functions, we embed $\{0,1\}^n$ into $\{0,1\}^{2n}$ as follows. For each vector $(a_1,a_2,\ldots,a_n)$, where $a_i\in\{0,1\}$, we identify it with $(a_1,a_2,\ldots,a_n, a_1+1,a_2+1,\ldots,a_n+1)$. Through this embedding, we can identify the elements of these two Boolean algebras as functions by assigning the values of the variables via 
{\footnotesize
\be
(x_1,x_2,\ldots,x_n,y_1,y_2,\ldots,y_n)=(a_1,a_2,\ldots,a_n, a_1+1,a_2+1,\ldots,a_n+1).
\ee}
For example, the function 
{\footnotesize
\be
f = x_2+x_1x_2+x_2x_3+x_1x_2x_3 = (x_1+1)x_2(x_3+1)
\ee}
can be identified with $f=y_1x_2y_3$.
\par
The following theorem is the main result.
\begin{theorem}\label{t2}
With the identification of  $B[\mathbf{x}]$ with $B[\mathbf{x}/\mathbf{y}]$ as described above, any Boolean functions $\{0,1\}^n\longrightarrow\{0,1\}$ can be represented by a polynomial of degree $\le n$ with at most $2^{n-1}$ terms.
\end{theorem}
\begin{proof}
For each vector $\mathbf{a}=(a_1,a_2,\ldots,a_n)\in\{0,1\}^n$, let $p_{\mathbf{a}}:=z_1z_2\cdots z_n$, where $z_i = x_i$ if $a_i = 1$, and $z_i=y_i$ if $a_i=0$. Then $p_{\mathbf{a}}(\mathbf{x}) = 1$ if $\mathbf{x}=\mathbf{a}$, and $0$ otherwise. For any function $f: \{0,1\}^n\longrightarrow\{0,1\}$, we have $f = \sum_{\mathbf{x}}f(\mathbf{x})p_{\mathbf{x}}$ (this is known as the {\it disjunctive normal form} of $f$). Furthermore, we have
\bea\label{e14}
\sum_{\mathbf{a}\in\{0,1\}^n}p_{\mathbf{a}} = 1.
\eea
Given $f$, we separate $\{0,1\}^n$ into two disjoint subsets
{\footnotesize
\be
S_0 := \{\mathbf{x}\in\{0,1\}^n\;|\;f(\mathbf{x})=0\}\;\mbox{and}\; S_1 := \{\mathbf{x}\in\{0,1\}^n\;|\;f(\mathbf{x})=1\}.
\ee}
Then by (\ref{e14}), we have 
\bea\label{e15}
f = \sum_{\mathbf{x}\in S_1}p_{\mathbf{x}} = 1+ \sum_{\mathbf{x}\in S_0}p_{\mathbf{x}}.
\eea
Now if $|S_1|\le 2^{n-1}$, then we use the first expression; and if $|S_1|>2^{n-1}$, then we can take the second expression. Q.E.D.
\end{proof}
\par
Though the above theorem guarantees the use of at most $2^{n-1}$ terms of degree $\le n$ monomials in representing a Boolean function, the expressions given by (\ref{e15}) are not necessary satisfactory. In applications, one can perform certain simplification process to simplify the expressions. Here we describe a simple procedure for reducing the number of monomial terms by combining terms using the relation $x_i+y_i=1,1\le i\le n$.
\par\vspace{1cm}
\begin{algorithm}\label{a1} Boolean Polynomial Function Reduction Algorithm
\par
INPUT: $f : \{0,1\}^n\longrightarrow\{0,1\}$.
\par
OUTPUT: A polynomial expression for $f$ of degree $\le n$ with at most $2^{n-1}$ terms of monomials in the expression.
\par
1. Choose a polynomial expression for $f$ according to Theorem \ref{t2}.
\par
2. Simplify the polynomial obtained in step 1 by combining pairs of terms differ only in one place using the relation $x_i+y_i=1$. This step can be repeated.
\par
3. Substitute in $x_i+1$ for $y_i$ in the polynomial obtained in step 2 and simplify module $2$. Then applying step 2 to combine terms if needed. 
\par
4. Compare the polynomials obtained in step 2 and step 3, and chose the one desired.
\end{algorithm}
\par
\begin{example}\label{ex1} Let $f: \{0,1\}^4\longrightarrow \{0,1\}$ be the function with output $1$ if the input sequence corresponds to an integer which is a sum of two squares ($0$ is allowed), and $0$ otherwise. Then $f$ outputs $1$ at the binary inputs correspond to the integers $0, 1,2,4,5,8,9,10,13$ (these integers are sums of two squares, i.e., $0=0^2+0^2$, $1=0^2+1^1$, $2=1^2+1^2$ etc.) and outputs $0$ for the binary inputs correspond to the integers $3,6,7,11,12,14,15$. The polynomial chosen in step 1 of Algorithm \ref{a1} is
{\footnotesize
\be
f  &=& 1+y_1y_2x_3x_4+y_1x_2x_3y_4+y_1x_2x_3x_4+x_1y_2x_3x_4\\
  {} &{}& \quad +x_1x_2y_3y_4+x_1x_2x_3y_4+x_1x_2x_3x_4.
\ee}
In step 2, we combine the second term with the fourth term $y_1y_2x_3x_4+y_1x_2x_3x_4=y_1x_3x_4$, the third with the seventh $y_1x_2x_3y_4+x_1x_2x_3y_4=x_2x_3y_4$, and the fifth with the eighth $x_1y_2x_3x_4+x_1x_2x_3x_4=x_1x_3x_4$, to get
\be
f = 1+y_1x_3x_4+x_2x_3y_4+x_1x_3x_4+x_1x_2y_3y_4.
\ee
It can be further simplified by combining the second and the fourth terms:
\be
f = 1+x_3x_4+x_2x_3y_4+x_1x_2y_3y_4.
\ee
\end{example}
%%%%%%%%%%%%%%%%%%%%%%%%%%%%%%%%%%%%%%%%%%%%%%%%%%%%%%%%%%%%%%%%%%%%%%%%%%%%%%%%%%%%%%%%%%%%%%%%%%%%%%%%%%%%%%%%%%%%%%%%%%%
\section{A Comparison Example}
\par
In this section, we compare our method with the result in \cite{Ozt09} using the example therein on restricted prime functions. The restricted prime functions $p_4$ and $p_5$ ($p_n$ counts the primes in $\{0, 1,\ldots, 2^n-1\}$) are explicitly given  in \cite{Ozt09} using sign representation polynomials with $7$ and $17$ terms respectively. Using our algorithm, we can represent these two functions with $4$ and $6$ terms respectively (the correctness of these polynomial representations can be checked by substitute in the binary sequences for the integers):
{\footnotesize
\begin{gather*}
p_4 =  x_1x_2x_4+x_1x_3x_4+y_1x_2x_4+y_1y_2x_3,\\
p_5 = y_1y_2y_3x_4+y_1y_2x_3x_5+x_2x_3y_4x_5+x_1y_2y_3x_5+x_1x_3x_4x_5+y_1x_2y_3x_4x_5.
\end{gather*}}
We give one more formula for these restricted prime functions and report the other comparison detail in Table 1. The function $p_6$ requires $39$ terms using sign representation in \cite{Ozt09}, using our method, it can be represented by the following polynomial with $11$ terms:
{\footnotesize
\be
p_6  &=&  y_1y_2y_3y_4x_5+y_1y_2y_3x_4x_6+y_2x_3y_4x_5x_6+y_1x_3x_4y_5x_6\\
{}&{}&\;+y_1x_2y_3y_4x_6+y_1x_2x_4x_5x_6+ x_1y_3x_4y_5x_6+ x_1y_2x_3x_5x_6\\
{}&{}&\qquad +x_1x_3y_4x_5x_6+ x_1y_2x_3y_4y_5x_6+x_1x_2x_3x_4y_5x_6.
\ee}
\par\vspace{-0.3cm}
In Table 1, the third column reports the results of running Algorithm \ref{a1} once for each case. In some cases, as seen in Example \ref{ex1}, it is possible to simplify the representations further by repeating step 2 and/or step 3 in the algorithm, but we decided to leave them as they are for a uniform presentation. All computations were done on a Dell laptop with a core due processor of $3.06$ GHz and $3.5$ GB RAM using MAPLE 11. 
\par\vspace{-0.3cm}
\begin{table}[h]
  \begin{center}
  {\renewcommand{\arraystretch}{1.4}\small 
  \begin{tabular}{|c|l|l|}
    \hline
    Restricted prime & \# Monomials used  & \# Monomials used by\\
    Functions & in \cite{Ozt09} (\% of full set) & our method (\% of full set) \\\hline
    $p_7$ & $82\; (64.06\%)$ & $23\; (17.97\%)$  \\
    $p_8$ & 147\; (57.42\%) & $38\; (14.84\%)$ \\ 
    $p_9$ & $315\; (61.52\%)$ & $66\; (12.89\%)$ \\ 
    $p_{10}$ & $633\; (61.82\%)$ & $115\; (11.23\%)$ \\
    $p_{11}$ & $1259\; (61.47\%)$ & $202\; (9.86\%)$ \\
    $p_{12}$ & (Not reported)      & $366\; (8.93\%)$ \\\hline 
 \end{tabular}
  }
 % \label{tab: rpf}
  \vskip 0.3cm
  \caption[Legend of variable names.]{\footnotesize Comparison table. The computing time for $p_{12}$ is $37.34$ sec.}
    \end{center}
\end{table}  
\par\vspace{-0.8cm}
%%%%%%%%%%%%%%%%%%%%%%%%%%%%%%%%%%%%%%%%%%%%%%%%%%%%%%%%%%%%%%%%%%%%%%%%%%%%%%%%%%%%%%%%%%%%%%%%%%%%%%%%%%%%%%%%%%%%%%%%%%%%%%%%%%%%%%%%%%%%%%%%%%%%%%%%%%
\section{Concluding Remarks}
\par
 We have proposed a novel method for representing a Boolean function using a polynomial. Our basic idea is to introduce extra variables to gain more flexibility in the representations. In Theorem \ref{t1}, we showed that the Boolean algebra $B[\mathbf{x}/\mathbf{y}]$ we have constructed is isomorphic to the classical Boolean algebra $B[\mathbf{x}]$,  so we can identify Boolean functions $\{0,1\}^n\longrightarrow\{0,1\}$ with polynomials of $B[\mathbf{x}/\mathbf{y}]$. In contrast to the uniqueness of the polynomials in $B[\mathbf{x}]$, we can represent a polynomial in $B[\mathbf{x}/\mathbf{y}]$ in different ways and thus make it possible to reduce the terms used. Our main theorem, Theorem \ref{t2} states that each Boolean function can be represented by a polynomial in $B[\mathbf{x}/\mathbf{y}]$ with at most $2^{n-1}$ monomial terms of degree at most $n$. In practice, one can apply Algorithm \ref{a1}, which is straight forward and easy to implement, to further reduce the number of terms used and thus reducing the number of input lines needed for the corresponding neuron. 
\par
 We compared our method with the best known deterministic algorithm for representing a Boolean function using a sign representation polynomial reported in \cite{Ozt09}. Table 1 shows that when applies to the restricted prime functions, our method produces significant improvements. However, as explained in the introduction, in general, each of the approaches examined here (including our proposed approach) has its own advantage, so in applications, one needs to choose the method which is appropriate for the type of applications under consideration, or to use a combined approach to achieve the best result. 
\par
Future work includes further theoretical analysis of the proposed new approach and its relationship with the existing approaches, in particular its relationship with sign representations.  

%%%%%%%%%%%%%%%%%%%%%%%%%%%%%%%%%%%%%%%%%%%%%%%%%%%%%%%%%%%%%%%%%%%%%%%%%%%%%%%%%%%%%%%%%%%%%%%%%%%%%%%%%%%%%%%%%%%%%%%%%%%%%%%%%%%%%%%%%%%%%%%%%%%%%%% 

%%%%%%%%%%%%%%%%%%%%%%%%%%%%%%%%%%%%%%%%%%%%%%%%%%%%%%%%%%%%%%%%%%%%%%%%%%%%%%%%%%%%%%%%%%%%%%%%%%%%%%%%%%%%%%%%%%%%%%%%%%%%%%%%%%%%%%%%%%%
\end{document}